\documentclass[12pt]{article}
\usepackage{geometry}                
\usepackage{graphicx}
\usepackage{epstopdf}
\usepackage{amsmath,amsfonts,amsthm,amssymb}
\usepackage{MnSymbol}

\newtheorem{theorem}{Theorem}
\newtheorem{proposition}[theorem]{Proposition}
\newtheorem{lemma}[theorem]{Lemma}

\theoremstyle{definition}
\newtheorem{definition}[theorem]{Definition}
\newtheorem{observation}[theorem]{Observation}
\newtheorem{example}[theorem]{Example}

\newtheorem{rem}[theorem]{Remark}
\usepackage{enumerate}

\usepackage[colorlinks]{hyperref}
\usepackage[T1]{fontenc}
\usepackage[utf8]{inputenc}
\usepackage{authblk}
\newcount\Comments  
\newcommand{\kibitz}[2]{\ifnum\Comments=1{\color{#1}{#2}}\fi}

\newcommand{\N}{\mathbb N}
\newcommand{\R}{\mathbb R}
\newcommand{\leaf}{\mathrm{leaf}}

\newcommand{\F}{\mathcal{F}}

\title{Golden games } 
\author[1]{Urban Larsson \thanks{urban031@gmail.com, partially supported by an Aly-Kaufman fellowship.}}
\author[2]{Yakov Babichenko \thanks{yakov@ie.technion.ac.il}}
\affil[1]{National University of Singapore, Singapore}
\affil[2]{Technion--Israel Institute of Technology, Haifa, Israel}

\date{}                                           

\begin{document}

\maketitle
\begin{abstract}
We consider extensive form win-lose games over a complete binary-tree of depth $n$ where  players act in an alternating manner.
We study arguably the simplest random structure of payoffs over such games where 0/1 payoffs in the leafs are drawn according to an i.i.d. Bernoulli distribution with probability $p$. Whenever $p$ differs from the golden ratio, asymptotically as $n\rightarrow \infty$, the winner of the game is determined. In the case where $p$ equals the golden ratio, we call such a random game a \emph{golden game}. In golden games the winner is the player that acts first with probability that is equal to the golden ratio. We suggest the notion of \emph{fragility} as a measure for ``fairness'' of a game's rules. Fragility counts how many leaves' payoffs should be flipped in order to convert the identity of the winning player. Our main result provides a recursive formula for asymptotic fragility of golden games. Surprisingly, golden games are extremely fragile. For instance, with probability $\approx 0.77$ a losing player could flip a single payoff (out of $2^n$) and become a winner. With probability $\approx 0.999$ a losing player could flip 3 payoffs and become the winner.
\end{abstract}

\section{Introduction}
Random games are games where the payoffs to the players are realized in accordance with a random process. 
In this paper we study a basic class of extensive form random win-lose games, denoted $G_n=G_n(p)$, over complete binary trees of depth $n\in \N_0$ with alternating moves; the payoffs are drawn Bernoulli i.i.d. in each leaf, i.e. each leaf has payoff ``1'' with probability $p$, and otherwise ``0''. The players are called Player~1 and Player~2, and we use the convention that Player~2 makes the last move in the game. In particular, Player~1 starts if the depth of the tree is even, and otherwise Player~2 starts.  The players observe the realization of $G_n$ before playing it.\footnote{That is, the realization of $G_n$ is an extensive-form complete-information zero-sum game.} Player~1 wins if the value of the game is ``1'' and otherwise Player~2 wins.

We focus on large random games when the depth $n$ of  $G_n$ grows to infinity. Our goal is to gain understanding for the following questions: 
\begin{enumerate}
    \item What is the asymptotic value (i.e., equilibrium outcome) of such a game?
    \item Can such a simple random process generate an extensive form game that is ``hard to play'' or repeatedly has ``fair rules''?
\end{enumerate}
 
  We denote by $\varphi = \frac{\sqrt{5}-1}{2}$ the golden ratio.  Regarding the first question we observe in Proposition~\ref{pro:p} the following: If $p>\varphi$, asymptotically, Player~1 wins the game with probability $1$. If $p<\varphi$,  asymptotically Player~2, wins the game with probability $1$. For $p=\varphi$ the value of the game remains uncertain (asymptotically) and the \emph{starting player} wins with probability $\varphi$ (recall the starting player is Player~1 if and only if the tree has even depth). Henceforth, we refer to random games with Bernoulli parameter that is equal the golden ratio ($p=\varphi$) as \emph{golden games}.

The second question requires more clarification in the interpretation of ``hard to play games'' and ``games with fair rules''. Regarding the notion of ``hardness'', it is problematic to apply the standard computational models since the input size of the problem is of size $2^n$. Moreover, the input does not admit succinct representation. Comparable to the input size $2^n$ obviously a backward induction procedure determines the value and a winning strategy for the winner in $2^n$ steps. Another standard complexity model is the query complexity. Here, it is not hard to show exponential hardness of golden games.

Regarding the notion of ``fairness'', it is problematic to state that some win-lose game is fair, because one player can guarantee winning. Nevertheless, it is reasonable to study relaxed notions of ``fairness'' that either assume boundedly rational play or, as we suggest, consider near-by games with very similar rules (payoffs) to argue that a very similar game may flip the identity of the winning player. 

We suggest a simple measure of \emph{game  fragility} to capture `fairness' and `hardness' in our specific scenario. Fragility relies on the notion of Hamming distance only. A game can be viewed as an element of $\{0,1\}^{2^n}$. The set of binary-tree alternating-move games can be partitioned into $W_1 \cupdot W_2$, when $W_i \subset \{0,1\}^{2^n}$ denotes the set of games where Player~$i$ wins. We define the fragility of a game $G\in W_i$ (i.e., a game where player $i$ wins) to be the Hamming distance of $G$ from the set $W_j$ for $j\neq i$. Namely, how many payoffs  shall the losing player switch in order to convert the value of the game. Equivalently, a game is $k$-fragile if the losing player can flip $k$ payoffs to convert the value of the game in favor of him. 
Intuitively, a game is fragile if it has a low `$k$-number'. 

Our main result shows that golden games are surprisingly fragile. Table \ref{tab:frag} presents the probability of a golden game to be $k$-fragile for asymptotically large depth.

\begin{table}[ht!]\label{tab:frag}
\caption{Approximate limiting probability of $d$-fragility for golden games.}\vspace{2 mm}
\begin{tabular}{c|c|c|c|c|c|}
\cline{2-6}
$d$                          & 1     & 2     & 3    & 4    & 5    \\ \cline{2-6} 
$\Pr [d$-fragility$]$ & 0.773 & 0.972 & $0.999$ & $1-5.57\times 10^{-5}$ & $1-6.98\times 10^{-7}$ \\ \cline{2-6} 
\end{tabular}
\end{table}

We find the fact that the probability of $1$-fragility is not asymptotically $0$ quite surprising. Recall that the number of leaves of $G_n$ is $2^n$. In most cases, the losing player can switch a single payoff to convert the value of the game in favour of him. In other words, the probability that a golden game will be located precisely on the boundary between $W_1$ and $W_2$ is $\approx 0.773$.
Moreover, one can see from Table~\ref{tab:frag} that the probability of $d$-fragility converges to $1$ extremely fast. We also deduce a theoretical statement that supports this statement; see Proposition~\ref{pro:p}.
Our main result  (Theorem~\ref{thm:main}) is a recursive (in $d$) formula for the asymptotic probability of $d$-fragility for golden games. We also show that fragility is unique for golden games; see Proposition \ref{pro:only-golden}. Namely, for $p\neq \varphi$, the asymptotic probability of a constant fragility equals $0$. 

\subsection{Related literature}

Random games have been used to rationalize mixed Nash equilibria (Harsanyi \cite{Harsanyi}), to study their generic properties in normal-form games, and as a tool to understand the complexity of the set of Nash equilibria. 
Most literature on random games focuses on normal form games and studies properties such as the expected number of Nash equilibria \cite{McLen05,BM05}, the distribution of pure Nash equilibria \cite{Dresher70,Powers90,Stanford95,Papa95,RS,Taka08}, or the maximal number of equilibria \cite{McLen97}.

Similar to our setting, \cite{REGF} studies games over binary trees with alternating moves where the payoff at the leaves are drawn i.i.d. according to some distribution;  \cite{REGF} considers a more general scenario where payoff profiles are not necessarily zero-sum. Unlike our setting, \cite{REGF} consider continuous density distributions, whereas we focus on the most fundamental setting of Bernoulli random payoffs that in particular are zero-sum. In this paper we go one step forward. We not only try to characterize in which cases the equilibrium outcome of the game is certain or uncertain (i.e., the probability of the value converges to a Dirac measure or assigns positive probability to the two outcomes 0 and 1), but also we focus on the question how stable (or in contrary fragile) is the value of the game.

Another remotely related line of literature is on the robustness of classifiers in machine learning to errors in data. 
Szegedy et. al. \cite{SZSBEGF14} show that the state-of-art classifiers are unstable to small adversarial perturbations in the data. This founding has led to an extensive line of research in this direction.
Our result is of similar spirit in a different setting of extensive form games. In our terminology, the classifier is the value of the game, and small perturbation is captured by flipping small number of bits. Our result can be formulated as stating that the value of golden games is not robust to adversarial perturbation of the game (this property is unique for golden games). However, in our context this property is not necessarily interpreted as a shortcoming. 

Related games have also been studied in the context of the disjunctive sum operator in combinatorial game theory (CGT), discussed in various papers on scoring combinatorial games. In particular, ``knotting-unknotting game'' of Pechenik, Townsend, Henrich, MacNaughton, and Silversmith \cite{J} are binary games with 0-1 payoffs and with constant parity of length of play. Observe that the popular CGT normal-play convention is not immediately applicable here, but our model should be generalized using guaranteed scoring combinatorial game theory \cite{LNS, LNNS}, in which the normal-play convention is order embedded. Observe that in the CGT disjunctive sum model, the order of play in distinct components is not necessarily alternating, so the setting in this paper would require a vast generalization. Still our model invokes some questions such as: is fragility related to winning by playing approximately optimal, i.e. if you are not a perfect player, can you still win non-fragile games with high probability?

\section{Model and main results}
This paper, similarly to \cite{REGF} focuses on random \emph{extensive form} games. We consider games on complete binary trees of depth $n$ where moves are done in an alternating manner, and the two players are Player~1 and Player~2. More concretely, we assume that Player~2 acts last, i.e. Player~1 acts in all nodes of height $2m$ for $0 \leq 2m\leq n$. Correspondingly,  Player~2 acts at all nodes of height $2m+1$ for $1\leq 2m+1 \leq n$.


The 0/1 payoffs in the $2^n$ leaves are i.i.d. Bernoulli random variables with probability of success $p$, with the standard convention that 1 indicates win of Player~1.

This random game is denoted by $G_n=G_n(p) \in \{0,1\}^{2^n}$, and the value of $G_n$ is denoted by $V_n = V_n(p) \in \{0,1\}$, which in itself is a random variable. The fact that we consider the value of the game, assumes that players observe the realization of $G_n(p)$ before playing it. 

For the special case when $p = \varphi = \frac{\sqrt{5}-1}{2}\approx 0.618$ is the golden section we call $G_n$ a \emph{golden game}.

Our first result is a characterization of the asymptotic value of such a random game as a function of $p$. 

\begin{proposition}\label{pro:p} 
Consider $G_n(p)$, with $p\in[0,1]$. 
\begin{itemize}
    
    \item For $p>\varphi$, then $\lim_{n\rightarrow \infty} \Pr[V_n(p)=1]=1$
    
      \item For $p<\varphi$, then $\lim_{n\rightarrow \infty} \Pr[V_n(p)=0]=1$
    
    \item For golden games, the starting player wins with probability $\varphi$, i.e. for all $n\in \N_0$,  $\Pr[V_{2n}(\varphi)=1]=\Pr[V_{2n+1}(\varphi)=0)]=\varphi$
\end{itemize}
\end{proposition}

Namely, the proposition states that for $p$ above the golden ratio, Player~1 typically wins such a random game. For $p$ below the golden ratio, Player~2 typically wins such a random game. 
For golden games the identity of the winner remains uncertain even for large values of $n$; the probability of Player~1 being the winner depends only on the parity of $n$.

 Now we proceed to our main result on the fragility of golden games.
			
For $i\in\{1,2\}$, denote by $W^i_n\subset \{0,1\}^{(2^n)}$ the set of all games $G_n$ with value $2-i$ (i.e., those games where Player~$i$ wins).\footnote{This notion does not depend on that $G_n$ is defined as a random game.}
\begin{definition}
Let $d\in \N$. A game $G_n\in W^i_n$ is $d$-fragile, denoted $G_n\in\F(d)$, if the Hamming distance of $G_n$ from the set $W^{3-i}_n$ is at most $d$.\footnote{Note that this definition does not use that $G_n$ is a random game.}
\end{definition}

Note that, for any $p$, $G_n(p)$ is trivially $2^n$-fragile. 

\begin{definition}
 Let $F_n(d):=\Pr[G_n(\varphi)\in\F(d)]$ be the probability that a golden game of depth $n$ is $d$-fragile. \end{definition}

\begin{example}\label{ex:fra}
The 1-fragility of the trivial game $G_0(\varphi)$ is $F_0(1)=1$, because if the single payoff is 0 (probability $\varphi$) then Player~1 can win by flipping it, and if the single payoff is 1 (probability $\varphi^2$), then Player~2 wins by flipping it. Further, the probability of 1-fragility of $G_1(\varphi)$ is
\begin{align}
F_1(1) &= \varphi \Pr[G_1\in \F(1)\mid V_1=0] + \varphi^2  \Pr[G_1\in \F (1)\mid V_1=1]\notag\\
&=\varphi (1-\Pr[G_1\not\in \F(1)\mid V_1=0]) + \varphi^2\cdot 1\notag\\ 
&=\varphi (1-\varphi^3)+\varphi^2=1-\varphi^4\label{eq:vp4}
\end{align}
For example $\Pr[G_1\not\in \F(1)] = \varphi^4$, and since we condition on that Player~2 wins $G_1$, according to Proposition~\ref{pro:p} third item, the first term in \eqref{eq:vp4} is correct. 

\end{example}

\begin{theorem}\label{thm:main}
The asymptotic $d$-fragility of a golden game exists and it satisfies for all $d\in\N$, 
\begin{align*}
    \lim_{n\rightarrow\infty} F_n(d)
    &=1-\varphi \xi_d -\varphi^2\xi_d^2,
\end{align*}
where $\xi_d\in (0,1)$ is the smaller root of the second degree polynomial $\varphi^3H(d)-x+2\varphi^2x^2$, where $H(1)=1$, and for $d>1$, 
$$H(d) = 1-\sum_{r+s=d}(1-\xi_r^2)(1-\xi_s^2)+\sum_{r+s=d-1}(1-\xi_r^2)(1-\xi_s^2)$$

\end{theorem}

\begin{rem}
In fact, the proof of Theroem~\ref{thm:main} provides more information about the problem than just asymptotic fragility. It also provides  precise answers for notions as asymptotic fragility conditional on the identity of the winner, and conditional on the parity of $n$. Moreover, recursive formulas as a function of the depth $n$ (rather than asymptotic) are deduced as well. 
For clarity of presentation we relegate these additional insights to the proof.
\end{rem}

Table \ref{tab:frag} demonstrates the calculations of Theorem \ref{thm:main} for $d\leq 5$.
Although it can be clearly seen from the table that the probability of $d$-fragility tends to 1 very quickly, this statement requires a theoretical proof. Indeed we have the following proposition.

\begin{proposition}\label{pro:d->inf}
We have
$\lim_{d\rightarrow \infty} F(d)=1$.\footnote{Note that this is a double-limit statement because $F(d)$ in itself is an asymptotic probability. If we convert the limits, the statement remains correct but trivial. Indeed for every fixed $n$, when $d \rightarrow \infty$ from sufficiently large $d$ we have $d\geq 2^n$ and obviously by flipping all payoffs the game in favour of one player makes him win.} 
\end{proposition}

For $p > \varphi$, Player~1 typically wins a large random game; see Proposition~\ref{pro:p}. Are these games fragile? The following proposition provides a negative answer to this question.

We denote by $F_n(d,p):=\Pr[G_n(p) \text{ is } d \text{-fragile}]$ the probability of a random game with Bernoulli parameter $p$ (not necessarily a golden game) to be $d$-fragile.

\begin{proposition}\label{pro:only-golden}
For every $p\neq \varphi$ and every fixed $d$ we have $\lim_{n\rightarrow \infty} F_n(d,p)=0$. 
\end{proposition}

This proposition stands in a sharp contrast to golden games where this limit approaches 1 very quickly.

\section{Proofs}
In this section, we will restate the results and prove them. As before, let $G_n=G_n(p)$ denote an instance of the 2-player random game with starting position the root of a binary tree of rank $n$, and with terminal positions the $2^n$ leaves.\\

\noindent {\bf Proposition~\ref{pro:p}.} {\it Consider $G_n(p)$, with $p\in[0,1]$. 
\begin{itemize}
    \item[(i)] For $p>\varphi$, then $\lim_{n\rightarrow \infty} \Pr[V_n(p)=1]=1$.
    
    \item[(ii)] For $p<\varphi$, then $\lim_{n\rightarrow \infty} \Pr[V_n(p)=0]=1$.
    
    \item[(iii)] For $p=\varphi$, the starting player wins with probability $\varphi$, i.e. for all $n\in \N_0$,  $\Pr[V_{2n}(\varphi)=1]=\Pr[V_{2n+1}(\varphi)=0)]=\varphi$.  
\end{itemize}
}

\begin{proof}
We begin by proving item (i) and item (ii) is similar. 
Let $g(x)$ denote the probability that Player~2 wins $G_2(1-x)$. Then $g(x)=(1-(1-x)^2)^2$, and note that $x$ is the probability that Player~2 wins $G_0(1-x)$.

Similarly, if we let $g^n=g\circ g^{n-1}$, for all $n>0$, then $g^n(x)$ is the probability that Player~2 wins $G_{2n}(1-x)$. 
 
 Consider any fixed $0<\delta< \varphi$. We demonstrate that $g^n(\varphi^2+\delta)$ is increasing. Observe that $g(x)=x^4-4x^3+4x^2$, so it suffices to show that $x^4-4x^3+4x^2>x$ on the open interval $(\varphi^2,1)$, that is that $f(x)=x^3-4x^2+4x>1$. One can check that $f(1)=f(\varphi^2)=1$, and that a local maximum is at $x=2/3$, i.e $f$ is increasing on $(\varphi^2,2/3)$ and decreasing on $(2/3,1)$, which settles this part. 

Since $g$ is increasing, it must converge, and we are interested in a number $1\ge\gamma =\lim_n g^n(\varphi^2+\delta)$. Namely, we want to show that $\gamma =1$ if $\delta>0$. To this purpose, we find all the fixed points of the equation $x=(1-(1-x)^2)^2$, that is $x=0$ or we solve the third degree equation $1=4x-4x^2+x^3$, so that $x=1$, $x=\frac{3-\sqrt{5}}{2}$ or $x=\frac{3+\sqrt{5}}{2}$. The third solution is too large, the second one will be too small, by $\frac{3-\sqrt{5}}{2} = \varphi^2$ and $\delta>0$. The first solution will be our fixed point $\gamma =1$.

 For item (iii), the base cases are $\Pr[V_0=1]=\varphi$, and $\Pr[V_1=0]=1-\varphi^2=\varphi$. Namely player Player~2 cannot force a win of $G_1$ if and only if each child of $G_1$ has payoff ``0''. By induction, we assume that $\Pr[V_n=1]=\varphi$ if $n$ is even. Then $\Pr[V_{n+1}=0]=1-\varphi^2=\varphi$, by the previous argument. And reversely, by induction, we assume that $\Pr[V_n=0]=\varphi$ if $n$ is odd. Then $\Pr[V_{n+1}=1]=1-\varphi^2=\varphi$.
\end{proof}

The notion of fragility is conditioned on one of the player losing, and being able to convert the value.
In the proofs to come it turns out to be more convenient to work with the reverse concept `robustness', i.e. `you are winning and I am not able to convert the value of the game'.
\begin{definition}\label{def:alphabeta}
For a given fragility number $d\in\N$, and any $n\in\N$, let $\alpha_n(d)=\Pr [G_n\not\in \F(d)\mid V_n=1]$ and $\beta_n(d)=\Pr [G_n\not\in \F(d)\mid V_n=0]$. 
\end{definition}
That is, the $\alpha$ sequence concerns the conditional probability of Player~2 losing, and not being able to change the value of the game to a ``0'', by flipping at most $d$ payoffs; the $\beta$ sequence concerns the conditional probability of Player~1 losing, and not being able of changing the value to a ``1'', by flipping at most $d$ payoffs. 


For each fragility number $d\in\N$, define functions $z_0,z_1:\N_0\rightarrow \R$ by 
\begin{align}\label{eq:z}
z_0(d)&=\lim_{n\rightarrow\infty}(\varphi (1-\alpha_{2n}(d))+\varphi^2(1-\beta_{2n}(d))),
\end{align}
and
\begin{align}\label{eq:z}
z_1(d)&=\lim_{n\rightarrow\infty}(\varphi^2 (1-\alpha_{2n+1}(d))+\varphi(1-\beta_{2n+1}(d))),
\end{align}
if the limits exist. 

We will show that the limits exist, and that in fact $z(d)=z_0(d)=z_1(d)$. Clearly, if this holds, then $z(d)$ equals \emph{the limiting probability} that $G_n(\varphi)$ is $d$-fragile, i.e. $F(d) = z(d)$ is the limiting probability that the losing player can win by flipping at most $d$ payoffs. 

Thus, for each $d$, it suffices to find explicit values for $\lim\alpha_{2n}(d)$ and $\lim\alpha_{2n+1}(d)$ to find $F(d)$. We will show that these values exist and can be computed recursively by using roots of quadratic equations derived from similar roots on smaller fragility numbers. 

First, we prove the main result (Theorem~\ref{thm:main}) for the case of 1-fragility. This helps to outline the main ideas. Let $G^L_n$  denote the game that starts with the left child of $G_n$ as root (where the payoffs remain the same), and let $V_n^L$ denote the game value of $G_n^L$, and analogously for $G_n^R$. Thus $G^L_n$ and $G^R_n$ have depths $n-1$. The fragility of $G_n$ naturally depends on the fragilities of $G^L_n$  and $G_n^R$, with specified details as in the below proofs.

\begin{proposition}\label{pro:1fr}
The asymptotic probability of $1$-fragility of a golden game exists and it satisfies 
\begin{align*}
    F(1)
    &=1-\varphi \xi -\varphi^2\xi^2
\end{align*}
where $\xi\in (0,1)$ is the smaller root of the second degree polynomial $\varphi^3-x+2\varphi^2x^2$.
Moreover $\xi = \lim \alpha_{2n}=\lim\beta_{2n-1}\approx 0.309017$ and $\xi^2=\lim \alpha_{2n-1}=\lim \beta_{2n}\approx 0.0954915 $.
\end{proposition}
\begin{proof}
Consider 1-fragility. Clearly $\alpha_0=\beta_0=0$ (conditioning on that the other player wins the losing player can flip the single payoff and win), and we wish to prove that, for all $n>0$, 
\begin{itemize}
\item[(i)] $\alpha_{2n}=\varphi^3+2\varphi^2\alpha_{2n-1}$ and $\alpha_{2n-1}=\alpha_{2n-2}^2$ 
\item[(ii)] $\beta_{2n+1}=\varphi^3+2\varphi^2\beta_{2n}$ and $\beta_{2n}=\beta_{2n-1}^2$
\end{itemize}
It suffices to prove (i), and (ii) is similar.

For $G_i$ with $i>0$, we condition on that $V_i=0$, and motivate (i) by the equivalent recurrences, $\alpha_0=0$ and for $n>0$
\begin{align}\label{eq:2n}
\alpha_{2n}=\frac{\varphi^4+2\varphi^3\alpha_{2n-1}}{\varphi} 
\end{align}
and 
\begin{align}\label{eq:2n-1}
\alpha_{2n-1}=\frac{\varphi^2(\alpha_{2n-2})^2}{\varphi^2}
\end{align}
\vspace{1 mm}

\noindent Case $i=2n$: In this case, since Player~1 starts, we prove that $\alpha_{2n}$ is $$\Pr[G_{2n}\not\in\F(1)\mid V_{2n}=1]= \frac{\Pr[\text{Player~2 cannot change each child that has value 1}]}{\Pr[V_{2n}=1]}$$ By Proposition~\ref{pro:p} (iii), $\Pr[V_{2n}=1] = \varphi$, so the conditioning probability is correct. 

The first term in the numerator of \eqref{eq:2n}, 
\begin{align*}
\varphi^4&=\Pr[V_{2n-1}=1]^2
\end{align*}
corresponds to the probability $(\varphi^2)^2$ that both children  $G^L_{2n}$  and $G^R_{2n}$ have value 1 (again Proposition~\ref{pro:p} (iii)), in which case the value of $G_{2n}$, cannot be changed, since $d=1$; at most one payoff at leaf level can be flipped so at most one of $V^L_{2n}$  and $V^R_{2n}$ can change, but since $2n$ is even, Player~1 starts, and thus can choose $G^L_{2n}$ or $G^R_{2n}$ as appropriate. 

The second term in the numerator of \eqref{eq:2n} corresponds to the situation where exactly one child has value 1, and we consider the probability that this child cannot be changed (two ways). For this term, we use induction. Thus \eqref{eq:2n} is correct.\\

\noindent Case $i=2n-1$: In this case, since Player~2 is the starting player, $\alpha_{2n-1}$ is $$\Pr[G_{2n-1}\not\in\F\mid V_{2n-1}=1] = \frac{\Pr[\text{Player~2 cannot change one child}]}{\Pr[V_{2n-1}=1]}$$

By Proposition~\ref{pro:p}, $\Pr[V_{2n-1}=1] = \varphi^2$, so the conditioning probability is correct. Moreover,  the conditioning implies that both children have value 0, since Player~1 starts; i.e. $V^L_{2n-1}=V^R_{2n-1}=1$. Since the depth of the children is even, by Proposition~\ref{pro:p}, this probability is $\varphi^2$. By induction, the probability that neither child can be changed is $(\alpha_{2n-2})^2$. Thus \eqref{eq:2n-1} is correct.

Finally, by induction $\alpha_{2n-2} \le \varphi$ implies that the sequence $(\alpha_{2n})$ is decreasing, and thus $(\alpha_{2n-1})$ is also decreasing. Since $(\alpha_i)$ is bounded below (by 0), the limits $\alpha_{2n}$ and $\alpha_{2n-1}$ exist, and thus $F(1)=z(1)$, as in (\ref{eq:z}), is correctly defined, namely take $\xi=\lim\alpha_{2n}$ and $\xi^2=\lim\alpha_{2n-1}$, i.e. $a(1) = 1 - \xi $ and $b(1) = 1 - \xi^2$. 
\end{proof}

The generalization of this result, to $d$-fragility, relies on a counting argument of pairs of 2-partitions of natural numbers. Again, without loss of generality (by Proposition~\ref{pro:p}), in the proofs we will condition on that Player~2 loses.

The new case concerns $G_{2n}, n>0$. If both children have value 1, then Player~2 must change both to win a modified game; this is a special case of $d$-fragility, which relies on a combinatorial counting argument. (In the below proof, we will let $H(d, 2n)$ denote the probability that the value of at least one child of $G_{2n}$ cannot be changed, that is, that such an instance of $G_{2n}$ is not $d$-fragile.) 
\begin{lemma}\label{lem:counting}
Fix any $d\in\N$. Consider the set $S=S_d$ of all ordered pairs $(r, s)$, with $r+s=d$. Now study the set $R=R_d=\{(x,y)\mid x,y\in S, x\ne y\}$ of ordered pairs of such ordered pairs, except the pairs of identical ordered pairs. Let $T =T_d= \{(\min\{x,u\}, \min\{y,v\})\mid ((x,y),(u,v))\in R \}$. Then $T$ contains exactly all the pairs (r,s), with $1<r+s<d$. Moreover $|T|=(d-1)(d-2)/2$.
\end{lemma}
\begin{proof}
This is standard combinatorics.
\end{proof}
We illustrate Lemma~\ref{lem:counting} with an example.
\begin{example}\label{ex:d4}
Let $d=4$. Then $S=\{(1,3),(3,1),(2,2)\}$, \\ $R=\{\{(1,3),(3,1)\},\{(2,2),(3,1)\},\{(2,2),(1,3)\}\}$, and $T=\{(1,1),(1,2),(2,1)\}$.
\end{example}
Let us explain the idea of how we use Lemma~\ref{lem:counting} in a game tree $G$ of even depth. If $V^L =V^R=0$, then in case of fragility, Player~1 must flip a combination of $(s,r)$ payoffs, with $s+r\le d$, where $0<s<d$ leaves belong to the the sub game tree $G^L$, and where $0<r<d$ leaves belong to the sub game tree $G^R$. For example, if $d=4$ and $G$ is fragile with $(r,s)=(1,2)$, then it is also fragile for $(2,2)$ and $(1,3)$. In the other direction, if $G$ is fragile for both $(2,2)$ and $(1,3)$ then it is also fragile for (1,2) (and perhaps for $(r,s) = (1,1)$). In our special case of binary trees and 2-partitioning, by Lemma~\ref{lem:counting}, the inclusion-exclusion principle leads to a neat  succinct description, as stated in \eqref{eq:inclexcl} below. 

We restate the main theorem with some notation to be used in its proof.\\

\noindent{\bf Theorem~\ref{thm:main}.}
{\it The asymptotic probability of $d$-fragility of a golden game $F(d)$ satisfies, for all $d\in\N$, 
\begin{align*}
    F(d)
    &=1-\varphi \xi_d -\varphi^2\xi_d^2
\end{align*}
where $\xi_d\in (0,1)$ is the smaller root of the second degree polynomial $\varphi^3H(d)-x+2\varphi^2x^2$, where $H(1)=1$, and for $d>1$
\begin{align}\label{eq:inclexcl}
H(d) = 1-C(d)+C(d-1), 
\end{align}
where $C(1)=1$, and for $d>1$,
  \begin{align*}
    C(d)&=\sum_{r+s=d, \,0<r<d}(1-\xi_r^2)(1-\xi_s^2)
\end{align*}
}
\begin{proof}
The proof is similar to that of Proposition~\ref{pro:1fr}; it remains to justify the factor $H(d)$ in \eqref{eq:inclexcl}, and argue that the $x$-sequence is decreasing for any $d\ge 1$. 

Consider $d$-fragility, and recall Definition~\ref{def:alphabeta}. Clearly $\alpha_0=\beta_0=0$, and we wish to prove that, for all $n\ge 0$, 
\begin{itemize}
\item[(i)] $\alpha_{2n}=\varphi^3H(d,2n)+2\varphi^2\alpha_{2n-1}$ and $\alpha_{2n-1}=\alpha_{2n-2}^2$ 
\item[(ii)] $\beta_{2n+1}=\varphi^3H(d,2n+1)+2\varphi^2\beta_{2n}$ and $\beta_{2n}=\beta_{2n-1}^2$,
\end{itemize}

where, for all $i\in\N_0$, 
\begin{align}\label{eq:ninclexcl}
H(d,i) = 1-C(d,i)+C(d-1,i), 
\end{align}
and where, for all $i$, $H(1,i)=C(1,i)=1$, and otherwise 
  \begin{align*}
    C(d,i)&=\sum_{r+s=d, \,0<r<d}(1-\xi_{r,i}^2)(1-\xi_{s,i}^2),
\end{align*}
where for all $0<r<d$, by induction $\alpha_{n,r}$ satisfies the following recurrences; $\alpha_{0,r}=0$ and for $n>0$
\begin{align}\label{eq:d2n}
\alpha_{2n,r}=\frac{\varphi^4H(r,2n)+2\varphi^3\alpha_{2n-1,r}}{\varphi} 
\end{align}
and 
\begin{align}\label{eq:d2n-1}
\alpha_{r, 2n-1}=\frac{\varphi^2(\alpha_{r,2n-2})^2}{\varphi^2}
\end{align}
As before, by symmetry, it suffices to study the cases where we condition on that Player~2 loses, and thus we are concerned with the $\alpha$ sequence. Therefore, it suffices to prove that $\lim_{n\rightarrow\infty}H(d,2n) = H(d)$ exists, and that indeed $H(d,2n)$ counts the probability of non-fragility for the case of two 1-valued children. Let us follow the outline of the proof of the $d=1$ case, and detail where the differences occur.\\

\noindent Case $i=2n$ \eqref{eq:d2n}: In this case, since Player~1 starts, $\alpha_{2n}=\alpha_{r,2n}$ is $$\Pr[G_{2n}\not\in\F(d)\mid V_{2n}=1]= \frac{\Pr[\text{Player~2 cannot change each child that has value 1}]}{\Pr[V_{2n}=1]}$$ By Proposition~\ref{pro:p} (iii), $\Pr[V_{2n}=1] = \varphi$, so the conditioning probability is correct. 

The first term in the numerator of \eqref{eq:d2n}, 
\begin{align*}
\varphi^4&=\Pr[V_{2n}^L=V_{2n}^R=1]
\end{align*}
corresponds to the probability $(\varphi^2)^2$ that both children $G_{2n}^{\,L}$  and $G_{2n}^{\,R}$ have value 1 (again Proposition~\ref{pro:p} (iii)). Each term of $C(d,i)$ defines the probability of $G^L_i$ $r$-fragile and $G^R_i$ $s$-fragile. Here we use Lemma~\ref{lem:counting}, to see that indeed $H(d,i)$ is correctly defined. Namely, the over counting of $C(d,i)$ contains all terms $C(r,i)$, $0<r<d$, but with alternating signs. When we subtract $C(d-1,i)$ from $C(d,i)$, then all the smaller terms cancel, which proves the correctness of \eqref{eq:ninclexcl}. 


The second term in the numerator of \eqref{eq:d2n} corresponds to the situation where exactly one child has value 1, and this case is in direct analogy with the proof of Proposition~\ref{pro:p}.

Finally, by induction $\alpha_{2n-2} \le \varphi$ and $H(d,2n)$ trivially decreasing (the probability that both children are fragile increases with increasing $d$) imply that the sequence $(\alpha_{2n})$ is decreasing, and thus $(\alpha_{2n-1})$ is also decreasing. Since $(\alpha_i)$ is bounded below (by 0), the limits $\alpha_{2n}$ and $\alpha_{2n-1}$ exist, and thus $z$ is correctly defined (by using again Proposition~\ref{pro:p}).
\end{proof}

Next we show that non-golden games are asymptotically `robust', i.e. the probability of $G_n(p)\in\F$ tends to 0 for the $p$-Binomial distribution whenever $p \ne \varphi$. It suffices to study the case $p >\varphi $, by symmetry and  Proposition~\ref{pro:p}.\\ 

\noindent {\bf Proposition~\ref{pro:only-golden}.}
\emph{Consider $p$-Binomial distribution whenever $p \ne \varphi$. 
For any given $d\in\N$, the probability of $d$-fragility tends to 0.}
\begin{proof}
By symmetry, it suffices to study the case $p >\varphi $. By Proposition~\ref{pro:p}, Player~2 loses with  asymptotic probability~1. 

Proposition~\ref{pro:p} implies that for any probability $p'<1$, and any fragility number $d$ there is a tree $G_n(p)$, and a hight $2k < n$, such that the probability is $p'$ that no single node at hight $2k+1$, has value 1. The probability of $d$-fragility of the value of such a game $G_{n}$ is smaller than $1-p'$. Since $p'$ is arbitrary close to $1$, the result follows.
\end{proof}

\noindent{\bf Proposition~\ref{pro:d->inf}.}
\emph{The asymptotic probability of $d$-fragility satisfies $\lim_{d\rightarrow \infty} F(d)=1$.}
\begin{proof}
With notation as in Theorem~\ref{thm:main}, it suffices to prove that $H(d)\rightarrow 0$, as $d\rightarrow\infty$. Clearly $H(d)\rightarrow 0$ if and only if $\xi_d\rightarrow 0$. Hence, for a contradiction, suppose that there is a $1>\delta>0$ such that $\xi_d\downarrow \delta > 0$, when $d\rightarrow \infty$ (monotonicity implying convergence was already discussed in the proof of Theorem~\ref{thm:main}). $H(d)$ is the asymptotic probability that the losing player cannot change the value of each child, by flipping altogether at most $d$ payoffs. 

Hence, for each (sufficiently small) $\epsilon>0$, there is a $d$ such that 
\begin{align}
\delta &>\varphi^3(1-(1-\delta-\epsilon)^2) + 2\varphi^2\xi_d^2\label{eq:1}\\
&\ge  \varphi^3H(d) + 2\varphi^2\xi_d^2\label{eq:2}\\
&=\xi_d, \notag
\end{align}
which contradicts the definition of $\delta$. The inequality \eqref{eq:1} follows since both children cannot be changed, and by choosing $d$ sufficiently large so that $\xi_d$ is sufficiently close to $\delta$. Then inequality \eqref{eq:2} follows by definition of $H(d)$ for a sufficiently large $d$.  
\end{proof}

\end{document}